\documentclass{article}
\usepackage[utf8]{inputenc}
\usepackage[dvipdfmx]{graphicx}
\usepackage{fullpage}
\usepackage{amsthm}
\usepackage{amsmath}
\usepackage{amssymb}
\usepackage{mathtools}
\usepackage[colorlinks,citecolor=blue,bookmarks=true,linktocpage]{hyperref}
\usepackage{aliascnt}
\usepackage[numbered]{bookmark}
\usepackage[capitalise]{cleveref}
\usepackage[backend=biber,url=false,arxiv=abs,maxbibnames=100,isbn=false,sortcites=true]{biblatex}
\usepackage{cleveref}
\usepackage{xcolor}
\usepackage{physics}
\usepackage{todonotes}

\newtheorem{theorem}{Theorem}
\newtheorem{lemma}{Lemma}
\newtheorem{corollary}{Corollary}
\newtheorem{proposition}{Proposition}
\newtheorem{observation}{Observation}

\crefname{observation}{Observation}{Observations}

\newcommand{\symdif}{\mathbin{\triangle}}
\newcommand{\True}{{\tt true}}
\newcommand{\False}{{\tt false}}
\newcommand{\sat}{{\tt \#Sat}}
\newcommand{\naesat}{{\tt \#NAESat}}

\title{Finding One Local Optimum Is Easy---but What About Two?\thanks{
This work is partially supported by JSPS KAKENHI Grant Numbers
JP23K28034, 
JP24H00686, 
JP24H00697, 
JP22H03549, 
JP25K21273, 
JP25K03080, 
JP25K00136, 
JP20K19743, 
JP20H00605, and 
JP25H01114, 
and by JST CRONOS Japan Grant Number JPMJCS24K2. 
}}
\author{
Yasuaki Kobayashi\thanks{Hokkaido University. Email: \texttt{koba@ist.hokudai.ac.jp}}
\and
Kazuhiro Kurita\thanks{Nagoya University. Email: \texttt{kurita@i.nagoya.ac.jp}}
\and
Yutaro Yamaguchi\thanks{Osaka University. Email: 
\texttt{yutaro.yamaguchi@ist.osaka-u.ac.jp}}
}

\begin{document}

\maketitle

\begin{abstract}
    The class PLS (Polynomial Local Search) captures the complexity of finding a solution that is locally optimal and has proven to be an important concept in the theory of local search. 
    It has been shown that local search versions of various combinatorial optimization problems, such as \textsc{Maximum Independent Set} and \textsc{Max Cut}, are complete for this class. 
    Such computational intractability typically arises in local search problems allowing \emph{arbitrary weights}; in contrast, for \emph{unweighted} problems, locally optimal solutions can be found in polynomial time under standard settings. 
    In this paper, we pursue the complexity of local search problems from a different angle: We show that computing two locally optimal solutions is NP-hard for various natural unweighted local search problems, including \textsc{Maximum Independent Set}, \textsc{Minimum Dominating Set}, \textsc{Max SAT}, and \textsc{Max Cut}. 
    We also discuss several tractable cases for finding two (or more) local optimal solutions.
\end{abstract}

\section{Introduction}
\emph{Local search} is one of the most popular heuristics for solving hard combinatorial optimization problems, which is frequently used in both theory and practice~\cite{aarts2003local,2018heuristics,GK2003,MichielsAK07,MonienDT10,WS11}, as well as in numerous studies in the field of Artificial Intelligence~\cite{GarvardtGKM23,GaspersKOSS12,SelmanKC94,SelmanLM92,SunWWL0Y24,ZhangL05}.
A primitive implementation of local search methods is the following hill-climbing algorithm: Starting from an arbitrary initial (feasible) solution $X$, the algorithm iteratively replaces the current solution $X$ with a (strictly) improved solution that can be found in a local solution space $\mathcal N(X)$, called a \emph{neighborhood}, defined around the current solution~$X$, as long as such an improvement can be found.
A solution that cannot be further improved with this procedure is called a \emph{local optimal solution}.
A plausible reason local search methods are frequently used is that finding a local optimal solution seems much easier than finding a \emph{global optimal solution} since every global optimal solution is locally optimal by definition.
In fact, however, there are several known obstacles to finding a local optimal solution.

Even if a local improvement in a neighborhood can be found in polynomial time, it is not easy to find a local optimal solution in general.
For example, the well-known $k$-opt heuristics require exponentially many improvement steps to find a local optimal solution in \textsc{Traveling Salesperson Problem} for all $k \ge 2$~\cite{ChandraKT99,EnglertRV14}.
A similar phenomenon can be observed in several combinatorial problems with imposed neighborhood structures, such as \textsc{Weighted Max Cut} with FLIP neighborhood~\cite{SchafferY91}.
The complexity of finding a local optimal solution under this setting is captured by the class PLS, which is introduced by~\cite{JohnsonPY88}, and various local search problems (with particular neighborhood structures) are shown to be complete in this class~\cite{KomusiewiczM24,Krentel89,SchafferY91}.
In particular, Sch\"affer and Yannakakis~\cite{SchafferY91} showed that the problem of finding a stable cut in edge-weighted graphs is PLS-complete, where a cut $\{X, Y\}$ in an edge-weighted graph $G$ is \emph{stable} if $w(X, Y) \ge w(X \symdif \{v\}, Y \symdif \{v\})$ for all $v \in X \cup Y$, where $w(X, Y)$ denotes the total weight of edges between $X$ and $Y$.\footnote{Here, $\symdif$ denotes the symmetric difference of two sets.}
In other words, a cut is stable if it is locally optimal under FLIP neighborhood~\cite{SchafferY91}.
Komusiewicz and Morawietz recently showed that the problem of finding a locally optimal weighted independent set under 3-swap neighborhood is PLS-complete~\cite{KomusiewiczM24}, where an independent set $Y$ is a \emph{3-swap neighbor} of an independent $X$ if $Y$ is obtained from $X$ by exchanging at most three vertices, that is, $|X \symdif Y| \le 3$.
We would like to mention that both problems are polynomial-time solvable in \emph{unweighted} or \emph{polynomially weighted}\footnote{A polynomially weighted graph is a vertex- or edge-weighted graph such that all the weights are positive integers bounded above by a polynomial in the size of the graph.} graphs as local optima can be obtained with a polynomial number of improvement steps from an arbitrary initial solution. 

Another obstacle is the complexity of finding a local improvement in a large-scale neighborhood.
This is frequently discussed through the lens of parameterized complexity theory~\cite{BergBJW21,FellowsFLRSV12,GarvardtGKM23,GaspersKOSS12,GuoHNS13,KomusiewiczM22,Marx08,MarxS11,Szeider11}.
In this context, we are given an instance and its feasible solution $X$ of a combinatorial problem and asked to find a better feasible solution $Y$ in a parameterized neighborhood $\mathcal N_k(X)$.
For the parameterized local search version of \textsc{Independent Set}, given a graph $G$, an independent set $X$ of $G$, and an integer $k$, the neighborhood $\mathcal N_k(X)$ of $X$ is defined as the collection of independent sets $Y$ of $G$ with $|X \symdif Y| \le k$, and hence the goal is to find an independent set $Y \in \mathcal N_k(X)$ with $|Y| > |X|$.
This problem is known to be W[1]-hard when parameterized by $k$~\cite{FellowsFLRSV12}, meaning that it is unlikely to exist an $f(k)|V(G)|^{O(1)}$-time algorithm.
Similarly to this, the local search version of \textsc{Max Cut} is also intractable~\cite{GarvardtGKM23}.

These obstacles can be immediately circumvented as long as we seek to find a local optimal solution for \emph{unweighted} or \emph{polynomially weighted} combinatorial optimization problems under \emph{polynomial-time computable neighborhood}, where a neighborhood function $\mathcal N$ is polynomial-time computable if given a solution $X$ of an instance $I$, its neighborhood $\mathcal N(X)$ can be computed in polynomial time in $|I|$.
In this paper, we explore the complexity of local search from yet another perspective, particularly focusing on \textsc{(Unweighted) Independent Set}, \textsc{Max Cut} on unweighted multigraphs, and \textsc{Max SAT} on unweighted CNF formulas.
An independent set $S$ of a graph $G$ is said to be \emph{$2$-maximal} if it is (inclusion-wise) maximal and, for every $v \in S$ and distinct $u, w \notin S$, $(S \setminus \{v\}) \cup \{u, w\}$ is not an independent set of $G$, that is, it is locally optimal under 3-swap neighborhood.
For a CNF formula, a truth assignment $\alpha$ is said to be \emph{unflippable} if the number of satisfied clauses under $\alpha$ is not smaller than the number of satisfied clauses under the truth assignment obtained from $\alpha$ by flipping the assignment of each single variable.
As mentioned above, we can find a 2-maximal independent set of $G$, a stable cut of an unweighted graph, and an unflippable truth assignment for a CNF formula in polynomial time by straightforward hill-climbing algorithms.

This paper establishes several complexity results for finding \emph{multiple} local optimal solutions in unweighted combinatorial optimization problems.
We show that it is NP-complete to decide whether an input graph has at least two $k$-maximal independent sets for any fixed $k \ge 2$.
It is worth mentioning that the problem of enumerating all maximal independent sets in a graph can be done in polynomial delay\footnote{An enumeration algorithm runs in polynomial delay if the time elapsed between two consecutive events is bounded by a polynomial in the input size, where the events include the initiation and the termination of the algorithm, and the output of each solution.}~\cite{TsukiyamaIAS77,JohnsonP88}, meaning that our result cannot be extended to the case $k = 1$.
This hardness result immediately yields similar hardness results for $k$-maximal cliques and $k$-minimal vertex covers as well.
We also show that it is NP-complete to decide whether an input graph has at least two $k$-minimal dominating sets and $k$-minimal feedback vertex sets.
For \textsc{Unweighted Max Cut} and \textsc{Unweighted Max SAT}, we show that the problems of finding two local optima (i.e., stable cuts and unflippable truth assignments, respectively) are NP-hard.

Our techniques used in this paper are not particularly new: We use standard polynomial-time reductions from known NP-complete problems.
However, a more careful construction in the reduction is required.
In a standard reduction for proving the NP-hardness of a problem $P$, given an instance $I$ of an NP-hard problem, we construct an instance $I'$ of $P$ such that $I$ has a solution if and only if $I'$ has a solution. 
However, this argument does not work for the local search problems that we discuss in this paper: Every instance has at least one local optimum, which can be computed in polynomial time.
To deal with this issue, we need to carefully construct an instance $I'$ that has a unique ``trivial solution'' that is irrelevant to a solution of $I$ and a ``nontrivial solution'' that is relevant to a solution of $I$.
This would make the construction of $I'$ and its proof more involved compared with standard NP-hardness reductions.

We would like to mention that our work is particularly relevant to practical aspects of local search. 
From a practical point of view, finding multiple local optimal solutions is crucial as some of them can be significantly bad compared to the global ones, while a simple hill-climbing algorithm may find such a bad solution.
There are numerous approaches to avoid getting stuck in bad local optimal solutions, such as iterated local search~\cite{LourencoMS03} and multi-start local search~\cite{MartiRR13}, in which one aims to mitigate this issue by essentially finding multiple local optimal solutions.
Our results would shed light on theoretical obstacles for such approaches.

On the positive side, we give a polynomial-time algorithm for deciding whether an input graph $G$ has at least two $2$-maximal matchings (i.e., $2$-maximal independent sets in the line graph).
Our algorithm also works for finding two $k$-maximal matchings for any fixed $k \ge 1$.
Since we can solve the maximum matching problem in polynomial time, our result may not give an interesting consequence to this end.
However, we believe that the result itself is nontrivial and still an interesting case that we can find multiple local optimal solutions in polynomial time, in contrast with our hardness results of $k$-maximal independent sets.
We also give an efficient algorithm for finding multiple $k$-maximal independent sets in bounded-cliquewidth graphs.

\section{Preliminaries}

Let $G$ be a (multi)graph.
The vertex set and edge set of $G$ are denoted by $V(G)$ and $E(G)$, respectively.
For $v \in V(G)$, the set of neighbors of $v$ in $G$ is denoted by $N_G(v)$.

A vertex set $S \subseteq V(G)$ is an \emph{independent set} of $G$ if no pair of vertices in $S$ are adjacent in $G$.
An independent set $S$ of $G$ is said to be \emph{maximal} if there is no vertex $v \in V(G) \setminus S$ such that $S \cup \{v\}$ is an independent set of $G$.
The concept of maximality can be generalized as in the following way.
For $k \ge 1$, an independent set $S$ is \emph{$k$-maximal} if for every $X \subseteq S$ with $|X| \le k - 1$ and $Y \subseteq V(G) \setminus S$ with $|Y| \ge |X| + 1$, $(S \cup Y) \setminus X$ is not an independent set of $G$.
Clearly, every (inclusion-wise) maximal independent set is 1-maximal and vice versa.
Moreover, every $k$-maximal independent set is $k'$-maximal for $k' \le k$.
\begin{observation}\label{obs:2-maximal}
    Let $S$ be a maximal independent set of $G$ that is not $2$-maximal.
    Then, there are a vertex $v \in S$ and distinct neighbors $u, w \in N_G(v) \setminus S$ such that $(S \cup \{u, w\}) \setminus \{v\}$ is an independent set of $G$. 
\end{observation}

Let $D \subseteq V(G)$.
A vertex $v \in V(G) \setminus D$ is \emph{dominated by $D$} (or $D$ \emph{dominates} $v$) if $N_G(v) \cap D \neq \emptyset$.
We say that $D$ is a \emph{dominating set} of $G$ if every vertex in $V(G) \setminus D$ is dominated by $D$.
Moreover, when $D$ is simultaneously an independent set and a dominating set of $G$, we call it an \emph{independent dominating set} of $G$.
Let us note that a vertex set is a maximal independent set of $G$ if and only if it is an independent dominating set of $G$.

A \emph{cut} of $G$ is an unordered pair of vertex sets $\{X, Y\}$ such that $X \cup Y = V(G)$ and $X \cap Y = \emptyset$.
The (multi)set of edges between $X$ and $Y$ is denoted by $E(X, Y)$.
A cut $\{X, Y\}$ is said to be \emph{improvable} if there is a vertex $v \in  V(G)$ such that $|E(X \symdif \{v\}, Y \symdif \{v\})| > |E(X, Y)|$.
A cut that is not improvable is called a \emph{stable cut} of $G$.

Let $\varphi$ be a CNF formula.
In this paper, $\varphi$ may contain multiple identical clauses.
The set of variables in $\varphi$ is denoted by $V(\varphi)$.
Let $\alpha: V(\varphi) \to \{\True, \False\}$ be a truth assignment for $\varphi$.
A clause is said to be \emph{satisfied} under $\alpha$ if it is evaluated to $\True$ under $\alpha$.
Moreover, a clause is \emph{NAE-satisfied} under $\alpha$ if it is satisfied under both $\alpha$ and its complement $\overline{\alpha}$, where $\overline{\alpha}(x) \coloneqq \neg \alpha(x)$ for $x \in V(\varphi)$.
The numbers of satisfied and NAE-satisfied clauses under $\alpha$ are denoted by $\sat_\varphi(\alpha)$ and $\naesat_\varphi(\alpha)$, respectively. 
We may omit the subscript $\varphi$ when no confusion is possible.
For $x \in V(\varphi)$, we define a truth assignment $\alpha_x$ as: for $y \in V(\varphi)$,
\begin{align*}
    \alpha_x(y) = \begin{cases}
        \alpha(y) & y \neq x\\
        \overline{\alpha}(y) & y = x
    \end{cases}.
\end{align*}
In other words, $\alpha_x$ is obtained from $\alpha$ by flipping the assignment of $x$.
We say that $\alpha$ is (\emph{NAE}-)\emph{flippable} if the number of (NAE-)satisfied clauses under $\alpha$ is strictly smaller than that of satisfied clauses under $\alpha_x$ for some variable $x \in V(\varphi)$, that is, $\sat(\alpha_x) > \sat(\alpha)$ ($\naesat(\alpha_x) > \naesat(\alpha)$); otherwise it is (\emph{NAE}-)\emph{unflippable}.

\section{Local Search Graph Problems}
In this section, we investigate the intractability of finding multiple local optima for several graph problems.
More precisely, we show that for any fixed $k \ge 2$, the problem of determining whether an input graph $G$ has at least two $k$-maximal independent sets is NP-complete.
We also prove that similar hardness holds for other local search graph problems.
In contrast, we show several tractable cases for finding multiple $k$-maximal independent sets.

\subsection{NP-Hardness}

We start with the case $k = 2$ and then extend it to the general case $k \ge 2$.
\begin{theorem}\label{thm:nph:main}
    It is NP-complete to determine whether an input graph has at least two $2$-maximal independent sets.
\end{theorem}

To prove \Cref{thm:nph:main}, we perform a polynomial-time reduction from the following problem.
Given a graph $G$ and $X \subseteq V(G)$, \textsc{Maximal Independent Set Extension} asks whether $G$ has a maximal independent set $D$ such that $D \cap X = \emptyset$.
In other words, the vertex set $D$ is an independent dominating set of $G$ that has no vertices in $X$.
This problem is known to be NP-complete~\cite{CaselFGMS19,ConteT22}.
As observed in \cite{CaselFGMS19}, we can assume that $X$ is an independent set of $G$.
Moreover, without loss of generality, we assume that $G$ has no isolated vertices.

We construct a graph $H$ as follows.
Let $Y = V(G) \setminus X$, where $Y = \{y_1, \dots, y_t\}$.
Starting from $H \coloneqq G$, we add five vertices~$a, b, b', c, c'$ and add edges between $a$ and $b$, between $b$ and vertices in $X$, and $\{b, c\}, \{b, c'\}, \{b', c\}, \{b', c'\}$.
The vertices $b, b', c, c'$ form a 4-cycle.
For each $y_i \in Y$, we add two vertices $z_i$ and $z'_i$ that are adjacent to $y_i$, $b$, and $b'$.
We let $Z = \{z_i, z'_i : 1 \le i \le t\}$.
Moreover, we add three vertices $c_i, c'_i, y'_i$ and edges that form a 4-cycle with $y_i$ for each $y_i \in Y$.
\Cref{fig:npc} illustrates the constructed graph $H$.
\begin{figure}
    \centering
    \includegraphics[width=0.6\linewidth]{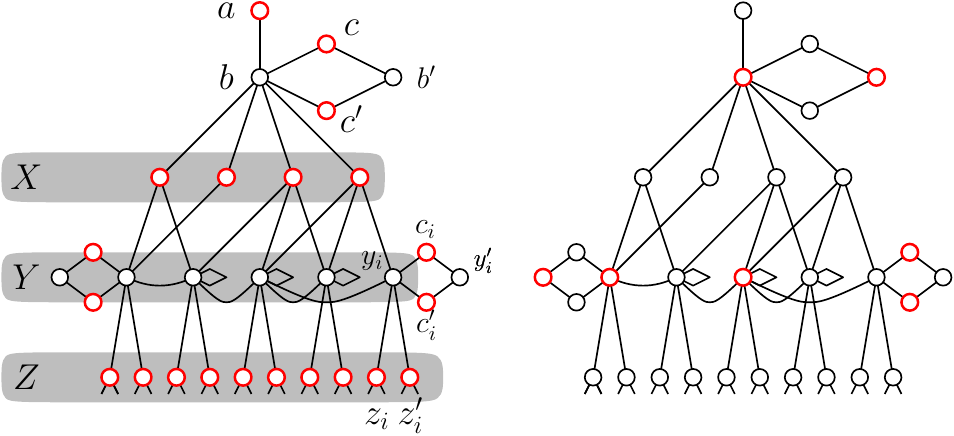}
    \caption{The left figure illustrates the graph $H$. The edges between $Z$ and $\{b, b'\}$ and some 4-cycles attached to vertices in $Y$ are omitted and simplified due to visibility. The red circles indicate the vertices in 2-maximal independent sets of $H$.}
    \label{fig:npc}
\end{figure}

The idea behind our construction is as follows.
Suppose that there is a 2-maximal independent set $S$ of $H$ with $b \in S$.
Since $S$ is an independent set, we have $a \notin S$ and $X \cap S = \emptyset$.
Due to the 2-maximality of $S$, $Y \cap S$ dominates all vertices in $X$: If $x \in X$ is not dominated by $Y \cap S$, $(S \setminus \{b\}) \cup \{a, x\}$ is an independent set of $H$.
To prove \Cref{thm:nph:main}, we need to ensure that $H$ has at most one 2-maximal independent set $S$ with $b \notin S$.
The other parts are used for this.

We then show that $G$ has an independent dominating set $D$ with $D \cap X = \emptyset$ if and only if $H$ has at least two $2$-maximal independent sets. 
To this end, we start with showing several observations.

\begin{observation}\label{obs:npc:incb}
    Let $S$ be a $2$-maximal independent set of $H$ with $b \in S$.
    Then, we have (1) $b' \in S$ and (2) $S \cap X = \emptyset$ and $S \cap Z = \emptyset$.
\end{observation}
\begin{proof}
    Since $b \in S$, we have $c, c' \notin S$.
    When $b' \notin S$, $S$ is not 2-maximal since $(S \cup \{c, c'\}) \setminus \{b\}$ is an independent set of $H$.
    The second statement follows as $b \in N_G(v)$ for all $v \in X \cup Z$.
\end{proof}

\begin{observation}\label{obs:npc:inca}
    Let $S$ be a $2$-maximal independent set of $H$ with $b \notin S$.
    Then, we have (1) $\{a, c, c'\} \subseteq S$ and (2) $S \cap Y = \emptyset$ and $X \cup Z \subseteq S$.
\end{observation}
\begin{proof}
    Since $b \notin S$, we have $a \in S$ due to the maximality of $S$.
    Suppose that $b' \in S$.
    Then, we have $c, c' \notin S$.
    This contradicts the $2$-maximality of $S$: $(S \cup \{c, c'\}) \setminus \{b'\}$ is an independent set of $H$.
    Thus, we have $b' \notin S$.
    By the maximality of $S$, we have $\{c, c'\} \subseteq S$.
    For the second statement, suppose to the contrary that $y_i \in S$ for some $y_i \in Y$.
    Since $z_i, z'_i \notin S$ and $b, b' \notin S$, it holds that $(S \cup \{z_i, z'_i\}) \setminus \{y_i\}$ is an independent set of $H$, contradicting the 2-maximality of $S$.
    Hence $S \cap Y = \emptyset$.
    As $b, b' \notin S$, we have $X \cup Z \subseteq S$ due to the maximality of $S$.
\end{proof}

The next observation follows from a similar discussion.

\begin{observation}\label{obs:npc:c4}
    Let $S$ be a $2$-maximal independent set of $H$.
    For each $1 \le i \le t$, either $\{y_i, y'_i\} \subseteq S$ or $\{c_i, c'_i\} \subseteq S$.
\end{observation}

\Cref{obs:npc:inca,obs:npc:c4} yields the following observation.
\begin{observation}\label{obs:npc:unique-a}
    Let $S_a = \{a, c, c'\} \cup X \cup Z \cup \{c_i, c'_i : 1 \le i \le t\}$.
    Then, $S_a$ is a 2-maximal independent set of $H$.
    Moreover, $H$ has no other 2-maximal independent set that contains $a$.
\end{observation}
\begin{proof}
    We show that $S_a$ is a 2-maximal independent set of $H$.
    Observe that for $v \in S_a$, each neighbor of $v$ is dominated by at least two vertices in $S_a$.
    Thus, by~\Cref{obs:2-maximal}, $S_a$ is 2-maximal.
    The uniqueness of $S_a$ immediately follows from \Cref{obs:npc:inca,obs:npc:c4}. 
\end{proof}

Now, we show that $H$ has a 2-maximal independent set $S$ with $S \neq S_a$ if and only if $G$ has a maximal independent set (or, equivalently, an independent dominating set) $D$ with $D \cap X = \emptyset$.

\begin{lemma}\label{lem:npc:forward}
    Suppose that $H$ has a 2-maximal independent set $S$ with $S \neq S_a$.
    Then, $G$ has an independent dominating set $D$ with $D \cap X = \emptyset$.
\end{lemma}
\begin{proof}
    Let $S$ be a 2-maximal independent set of $H$ with $S \neq S_a$ that maximizes $|S \cap Y|$.
    By \Cref{obs:npc:unique-a}, we have $a \notin S$ and hence, by~\Cref{obs:npc:incb}, we have $b \in S$, $S \cap X = \emptyset$, and $S \cap Z = \emptyset$.
    Let $D = S \cap Y$.
    Clearly, $D$ is an independent set of $G$ with $D \cap X = \emptyset$.
    We show that $D$ dominates all the vertices in $V(G) \setminus D$.
    
    Observe first that all the vertices of $Y \setminus D$ are dominated by $D$.
    To see this, suppose that there is vertex $y_i \in Y$ that is not dominated by $D$.
    Then, $y_i$ has no neighbor in $S \cap Y$, and hence $S' \coloneqq (S \setminus \{c_i, c_i'\}) \cup \{y_i, y_i'\}$ is an independent set of $H$.
    This independent set $S'$ is indeed 2-maximal as every neighbor of a vertex in $S'$ is dominated by at least one other vertex in $S'$. 
    This contradicts the assumption that $S$ maximizes $|S \cap Y|$ as $|S' \cap Y| > |S \cap Y|$.
    
    Suppose that $v \in X$ is not dominated by $D$.
    Then, $(S \cup \{a, v\}) \setminus \{b\}$ is an independent set of $H$, contradicting the 2-maximality of $S$.
    Therefore, $D$ is a dominating set of $G$.
\end{proof}

\begin{lemma}\label{lem:npc:converse}
    Let $D \subseteq Y$ be an independent dominating set of $G$.
    Then, $\{b, b'\} \cup \{y_i, y'_i : y_i \in D\} \cup \{c_i, c'_i : y_i \notin D\}$ is a 2-maximal independent set of $H$.
\end{lemma}
\begin{proof}
    Let $S = \{b, b'\} \cup \{y_i, y'_i : y_i \in D\} \cup \{c_i, c'_i : y_i \notin D\}$. 
    We claim that $S$ is a 2-maximal independent set of $H$.
    It is easy to verify that $S$ is an independent set of $H$.
    Since all vertices in $X \cup Z \cup \{a, c, c'\}$ are dominated by $b$ and $\{y_i, y'_i, c_i, c'_i\} \setminus S$ are dominated by $\{y_i, y'_i, c_i, c'_i\} \cap S$ for all $i$.
    Thus, $S$ is a maximal independent set of $H$.
    To see 2-maximality, by \Cref{obs:2-maximal}, it suffices to show that for each $v \in S$, there is at most one neighbor that is not dominated by $S \setminus \{v\}$.
    In the following, a vertex $v \in S$ is said to be \emph{stable} in $S$ if there is at most one neighbor that is not dominated by $S \setminus \{v\}$.

    Since the neighbors of $b'$ are also neighbors of $b$, it follows that $b'$ is stable in $S$.
    Moreover, $b$ is also stable as $b'$ dominates the vertices in $Z \cup \{c, c'\}$ and $D$ dominates the vertices in $X$.
    Note that $a$ is not dominated by other vertex in $S$.
    When one of $c_i, c'_i, y'_i$ is included in $S$, it is stable as pairs $\{c_i, c'_i\}$ and $\{y_i, y'_i\}$ dominate each other.
    Finally, if $y_i \in S$, it is stable since $c_i$ and $c'_i$ are dominated by $y'_i$, the vertices in $X \cup Z$ are dominated by $b$, and each neighbor $y_j$ in $Y$ is dominated by $c_j$.
    As all vertices in $S$ are stable, $S$ is 2-maximal.
\end{proof}

Hence, there are at least two 2-maximal independent sets in $H$ if and only if $G$ has an independent dominating set $D$ with $D \cap X = \emptyset$, completing the proof of \Cref{thm:nph:main}.

We can extend this proof to those for any fixed $k \ge 2$.
From an instance $(G, X)$ of \textsc{Maximal Independent Set Extension}, we construct the graph $H$ as above and convert it to a graph $H_k$ as follows.
We replace each vertex $v$ of $H$ with an independent set $M^v_k = \{v_1, \dots, v_{k - 1}\}$ of $k - 1$ vertices and add an edge between each vertex in $M^v_k$ and each vertex in $M^w_k$ if and only if $v$ and $w$ are adjacent in $H$.
The graph obtained in this way is denoted by $H_k$.
See~\Cref{fig:npc-k} for an illustration.
Note that $H = H_2$.

\begin{figure}
    \centering
    \includegraphics[width=0.5\linewidth]{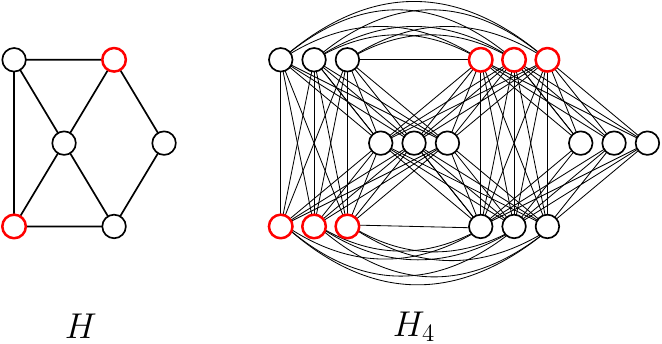}
    \caption{The graph $H_4$. The red circles indicate maximal independent sets.}
    \label{fig:npc-k}
\end{figure}

\begin{lemma}\label{lem:npc:k-ind:maximal}
    For any maximal independent set $S$ of $H$, $S_k \coloneqq \bigcup_{v \in S}M^v_k$ is a maximal independent set of $H_k$.
    Conversely, for any maximal independent set $S_k$ of $H_k$, $S \coloneqq \{v \in V(H): S_k \cap M^v_k \neq \emptyset\}$ is a maximal independent set of $H_k$.
    Moreover, $S_k \cap M^v_k \neq \emptyset$ implies that $M^v_k \subseteq S_k$ for any $v \in V(H)$.
\end{lemma}
\begin{proof}
    For the first statement, suppose that $S$ is a maximal independent set of $H$.
    Clearly, $S_k = \bigcup_{v \in S}M^v_k$ is an independent set of $H_k$.
    For each $v \in V(H) \setminus S$, it has a neighbor $w \in S$ due to the maximality of $S$.
    This implies that each vertex in $M^v_k \setminus S_k$ has a neighbor $M^w_k \cap S_k$.
    Thus, $S_k$ is a maximal independent set of $H_k$.
    
    For the second statement, suppose that $S_k$ is a maximal independent set of $H_k$.
    Then, $S = \{v \in V(H): S_k \cap M^v_k \neq \emptyset\}$ is an independent set of $H$.
    If $S_k \cap M^v_k \neq \emptyset$, we have $M^v_k \subseteq S$ due to the maximality of $S_k$.
    Suppose that $S \coloneqq \{v \in V(H): M^v_k \subseteq S_k\}$ is not a maximal independent set of $H$.
    Then, there is a vertex $v \in V(H) \setminus S$ such that $S \cup \{v\}$ is an independent set.
    This implies that $S_k \cap M^v_k = \emptyset$ and for any $w \in V(H)$ that is adjacent to $v$, $M^w_k \cap S_k = \emptyset$.
    Thus, $S_k \cup M^v_k$ is also independent in $H_k$, contradicting the maximality of $S_k$.
\end{proof}

The above lemma indicates that there is a bijection between the collection of maximal independent sets in $H$ and the collection of maximal independent sets in $H_k$.
Moreover, the vertices in $M^v_k$ are ``homogeneous'' in any maximal independent set in $H_k$.
Now, we claim that this bijection maps every $2$-maximal independent set of $H$ to a $k$-maximal independent set of $H_k$.
Let $S$ be a maximal independent set of $H$ and let $S_k = \bigcup_{v \in S}M^v_k$ be a maximal independent set of $H_k$, which is the image of $S$.
\begin{lemma}\label{lem:npc:k-ind:forward}
    If $S_k$ is a $k$-maximal independent set of $H_k$, then $S$ is a 2-maximal independent set of $H$.    
\end{lemma}
\begin{proof}
Suppose that $S$ is not 2-maximal.
Then, there are $v \in S$ and $u, w \in V(H) \setminus S$ such that $(S \setminus \{v\}) \cup \{u, w\}$ is an independent set of $H$.
Then $(S_k \setminus M^v_k) \cup (M^u_k \cup M^w_k)$ is an independent set of $H_k$, meaning that $S_k$ is not $k$-maximal as $|M^v_k| = k-1$ and $|M^u_k \cup M^w_k| = 2k-2 > k-1$ for $k \ge 2$.
\end{proof}

\begin{lemma}\label{lem:npc:k-ind:backward}
    If $S$ is a 2-maximal independent set of $H$, then $S_k$ is a $k$-maximal independent set of $H_k$.
\end{lemma}
\begin{proof}
Suppose that $S_k$ is a maximal independent set but is not a $k$-maximal independent set of $H_k$.
By~\Cref{lem:npc:k-ind:maximal}, either $M^v_k \subseteq S_k$ or $S_k \cap M^v_k = \emptyset$ for any $v \in V(H)$.
Then, there are at most $k - 1$ vertices $X \subseteq S_k$ and at least $|X| + 1$ vertices $Y \subseteq V(H_k) \setminus S_k$ such that $(S_k \setminus X) \cup Y$ is an independent set of $H_k$.
Since $S_k$ is maximal, $X$ is nonempty.
Let $v \in V(H)$ with $X \cap M^v_k \neq \emptyset$.
Since $M^v_k \subseteq S_k$ and $|M^v_k| = k - 1$, we can assume that $X = M^v_k$.
Similarly to \cref{obs:2-maximal}, each vertex in $Y$ is adjacent to a vertex in $M^v_k$ as $S_k$ is maximal.
As $|Y| \ge k$, $Y$ has at least one vertex in $M^u_k$ and at least one vertex in $M^w_k$ such that $u$ and $w$ are both adjacent to $v$ in $H$.
Then, $(S \setminus \{v\}) \cup \{u, w\}$ is an independent set of $H$, yielding that $S$ is not 2-maximal.
\end{proof}

By~\cref{lem:npc:forward,lem:npc:k-ind:backward}, we have the following.
\begin{theorem}\label{cor:npc:k-ind}
    For any fixed $k \ge 2$, it is NP-complete to determine whether an input graph has at least two $k$-maximal independent sets.
\end{theorem}

As an immediate corollary of \Cref{thm:nph:main,cor:npc:k-ind}, we have the NP-hardness of finding multiple local optima for the minimum vertex cover and the maximum clique problems.
Let $G$ be a graph.
A \emph{clique} is a pairwise adjacent vertex set in $G$.
A \emph{vertex cover} of $G$ is a vertex set such that its complement is an independent set of $G$.
A clique $S$ is \emph{$k$-maximal} if for $X \subseteq S$ with $|X| \le k - 1$ and $Y \subseteq V(G)\setminus S$ with $|Y| > |X|$, $(S \cup Y) \setminus X$ is not a clique in $G$.
Similarly, a vertex cover $S$ is \emph{$k$-minimal} if for $X \subseteq V(G) \setminus S$ with $|X| \le k - 1$ and $Y \subseteq S$ with $|Y| > |X|$, $(S \setminus Y) \cup X$ is not a vertex cover of $G$.
It is easy to see that a vertex set $S$ is a $k$-maximal independent set of $G$ if and only if it is a $k$-maximal clique in the complement graph of $G$.
Moreover, $S$ is a $k$-maximal independent set of $G$ if and only if $V(G) \setminus S$ is a $k$-minimal vertex cover of $G$.
Hence, the following corollary is immediate.
\begin{corollary}
    For any fixed $k \ge 2$, it is NP-complete to determine whether an input graph has at least two $k$-maximal cliques.
    Moreover, it is NP-complete to determine whether an input graph has at least two $k$-minimal vertex covers.
\end{corollary}

Similarly, a dominating set $D$ of $G$ is \emph{$k$-minimal} if for $X \subseteq V(G) \setminus D$ with $|X| \le k - 1$ and $Y \subseteq D$ with $|Y| > |X|$, $(D \setminus Y) \cup X$ is not a dominating set of $G$.
A slightly modified version of a well-known reduction from \textsc{Vertex Cover} to \textsc{Dominating Set} proves the hardness of finding multiple $k$-minimal dominating sets.

\begin{theorem}\label{thm:npc:dom}
    For any fixed $k \ge 2$, it is NP-complete to determine whether an input graph has at least two $k$-minimal dominating sets.
\end{theorem}
\begin{proof}
    Let $G$ be a graph without isolated vertices.
    For each two adjacent vertices $u, v \in V(G)$ with edge $e = \{u, v\}$, we add $k$ vertices $x^e_1, \dots, x^e_k$ and edges $\{u, x^e_i\}, \{x^e_i, v\}$, forming a path of length~2 between $u$ and $v$ for each $1 \le i \le k$.
    The graph constructed in this way is denoted by $H$.
    Now, we claim that $G$ has at least two $k$-minimal vertex covers if and only if $H$ has at least two $k$-minimal dominating sets.

    We first observe that for every $k$-minimal dominating set $S$ of $H$, it holds that $S \cap \{x^e_1, \dots, x^e_k\} = \emptyset$ for each $e \in E(G)$.
    To see this, suppose that $\{x^e_1, \dots, x^e_k\} \subseteq S$ for some edge $e = \{u, v\}$.
    Then, $(S \setminus \{x^e_1, \dots, x^e_k\}) \cup \{u\}$ is a dominating set of $H$, contradicting the $k$-minimality of $S$.
    Suppose next that $x^e_i \in S$ and $x^e_j \notin S$ for some $1 \le i, j \le k$.
    Then, $S$ must contain at least one of $u$ and $v$ as otherwise $x^e_j$ is not dominated.
    This implies that $S \setminus \{x^e_i\}$ is still a dominating set of $H$, which contradicts the minimality of $S$.
    Thus, we have $S \cap \{x^e_1, \dots, x^e_k\} = \emptyset$ for all $e \in E(G)$.
    
    It is easy to verify that for every vertex set $S \subseteq V(G)$, $S$ is a vertex cover of $G$ if and only if it is a dominating set of $H$. This bijection proves the claim, as $G$ has at least two $k$-minimal vertex covers if and only if $H$ has at least two $k$-minimal dominating sets.
\end{proof}

Finally, we show that finding $k$-minimal feedback vertex sets is hard.
A \emph{feedback vertex set} of a graph $G$ is a vertex subset $X \subseteq V(G)$ whose removal makes the graph acyclic (i.e., a forest).
The $k$-minimality for feedback vertex sets is defined analogously.

\begin{theorem}\label{thm:npc:fvs}
    For any fixed $k \ge 2$, it is NP-complete to determine whether an input graph has at least two $k$-minimal feedback vertex sets.
\end{theorem}
\begin{proof}
    The proof is almost identical to \Cref{thm:npc:dom}.
    Let $H$ be the graph constructed in the proof of \Cref{thm:npc:dom}.

    We observe that for every $k$-minimal feedback vertex set $S$ of $H$, it holds that $S \cap \{x^e_1, \dots, x^e_k\} = \emptyset$ for every $e \in E(G)$.
    To see this, suppose that $\{x^e_1, \dots, x^e_k\} \subseteq S$ for some $e = \{u, v\} \in E(G)$.
    Since all the cycles passing through $x^e_i$ in $H$ must contain both $u$ and $v$.
    This implies that $(S \setminus \{x^e_1, \dots, x^e_k\}) \cup \{u\}$ is a feedback vertex set of $H$.
    Suppose moreover that $x^e_i \in S$ and $x^e_j \notin S$ for some $i$ and $j$.
    Then, at least one of $u$ and $v$ are contained in $S$ as otherwise the 3-cycle formed by $u, v, x^e_i$ is not hit by $S$.
    Thus, $S \setminus \{x^e_i\}$ is a feedback vertex set of $H$.
    In any case, it contradicts the assumption that $S$ is $k$-minimal.

    Again, it is easy to verify that for every vertex set $S \subseteq V(G)$, $S$ is a vertex cover of $G$ if and only if it is a feedback vertex set of $H$.
    Thus, the theorem follows.
\end{proof}

\subsection{Polynomial-time algorithm for $k$-maximal matchings}
To complement \Cref{thm:nph:main}, we investigate the tractability of finding multiple local optima for a special case of the local search version of \textsc{Maximum Independent Set}.

Let $G$ be a graph.
A set of edges $M \subseteq E(G)$ is called a \emph{matching} of $G$ if every pair of edges in $M$ does not share their end vertices.
For an integer $k \ge 1$, a matching $M$ is said to be \emph{$k$-maximal} if for $X \subseteq M$ with $|X| \le k - 1$ and $Y \subseteq E(G) \setminus X$ with $|Y| > |X|$, $(M \cup Y) \setminus X$ is not a matching of $G$.
In this section, we give a polynomial-time algorithm for deciding whether an input graph $G$ has at least two $k$-maximal matchings for every fixed $k$.

Let $M$ be a matching of $G$.
In the following, we may not distinguish between an edge set of $G$ and the subgraph induced by them.
We say that a vertex $v$ is \emph{matched} in $M$ if there is an edge incident to $v$ in $M$.
A path $P = (v_1, \dots, v_\ell)$ is said to be \emph{$M$-alternating} if $v_1$ is not matched in $M$ and $\{v_i, v_{i+1}\} \in M$ for even~$i$.
An $M$-alternating path $P = (v_1, \dots, v_\ell)$ is said to be \emph{$M$-augmenting} if $v_\ell$ is not matched in $M$, that is, $\ell$ is even.
It is well known that a matching $M$ is maximum if and only if it has no $M$-augmenting paths (e.g., \cite{korte2018combinatorial}).
In particular, for a matching $M'$ with $|M'| < |M|$, there is a path component in $M' \symdif M$ that is $M'$-augmenting, as $M' \symdif M$ consists of a disjoint union of paths and cycles.

\begin{lemma}\label{lem:k-mm:char}
    Let $M$ be a matching in $G$ and let $k$ be a positive integer.
    Then, $M$ is $k$-maximal if and only if $G$ has no $M$-augmenting path of length at most $2k - 1$.
\end{lemma}
\begin{proof}
    Suppose that $G$ has an $M$-augmenting path of length at most $2k - 1$.
    Then, $M \symdif P$ is a matching with $|M \symdif P| > |M|$ such that $|M \cap P| \le k - 1$ and $|P \setminus M| > |M \cap P|$.
    This implies that $M$ is not $k$-maximal.

    Suppose that $M$ is not $k$-maximal.
    Then, there are $X \subseteq M$ and $Y \subseteq G \setminus M$ with $|X| = |Y| + 1 \le k$ such that $(M \setminus X) \cup Y$ is a matching of $G$.
    Let $M' = (M \setminus X) \cup Y$.
    As $|M'| > |M|$, there is a path component $P$ in $M \symdif M'$ that is $M$-augmenting.
    Moreover, each component in $H$ has at most $|X| + |Y| \le 2k-1$ edges.
    Hence, $P$ is an $M$-augmenting path of length at most $2k - 1$.
\end{proof}

Given a graph $G$, it is easy to determine if $G$ has at least two maximum matchings by using a polynomial-time algorithm for computing a maximum matching (see \cite{GabowKT01} for a more sophisticated algorithm).
Obviously, we can determine if $G$ has at least two $k$-maximal matchings in polynomial time when there are at least two maximum matchings in $G$.
Hence, in the following, we assume otherwise that $G$ has a unique maximum matching $M^*$.
A connected component in a graph is \emph{nontrivial} if it has at least one edge.

\begin{lemma}\label{lem:k-mm:almost-max}
    Let $M^*$ be a unique maximum matching of $G$.
    Suppose that $G$ has a $k$-maximal matching other than $M^*$.
    Then, $G$ has a $k$-maximal matching $M$ of size $|M^*| - 1$.
    Moreover, there is exactly one nontrivial component in $M \symdif M^*$, which is an $M$-augmenting path.
\end{lemma}
\begin{proof}
    Let $M$ be a $k$-maximal matching of $G$ with $M \neq M^*$.
    We assume that $|M| < |M^*|-1$ as otherwise we are done.
    Let $P^*$ be an $M$-augmenting path component in $M \symdif M^*$.
    Since $P^*$ is an $M$-augmenting path in $G$, by~\Cref{lem:k-mm:char}, $P^*$ contains more than $2k - 1$ edges.
    Let $M' = M^* \symdif P^*$.
    Clearly, $P^*$ is an $M'$-augmenting path, and $M'$ is a matching with $|M'| = |M^*| - 1$.
    Moreover, there is no $M'$-augmenting path other than $P^*$, as otherwise such an $M'$-augmenting path $P$ would imply a maximum matching $M' \symdif P$ distinct from $M^*$.
    Again, by~\Cref{lem:k-mm:char}, $M'$ is $k$-maximal.
\end{proof}

By~\Cref{lem:k-mm:almost-max}, it suffices to find a $k$-maximal matching $M$ with size $|M^*| - 1$.
Suppose that the $M$-augmenting path $P^*$ in the proof of \Cref{lem:k-mm:almost-max} has $2k + 1$ edges.
(Recall that each augmenting path has an odd number of edges.)
In this case, we try to check all the possibilities of paths $P$ with $|E(P)| = 2k + 1$ such that $M^* \symdif P$ is a $k$-maximal matching of $G$.
This can be done in polynomial time, as $k$ is fixed.
Suppose otherwise that $P^*$ has at least $2k + 3$ edges.
In this case, we claim that the above algorithm also finds a $k$-maximal matching of $G$ if it exists.
To see this, let $P$ be a subpath of $P^*$ containing one of the end vertices of $P^*$ such that $|E(P)| = 2k + 1$.
Then, $M \coloneqq M^* \symdif P$ is a matching of $G$ with $|M| = |M^*| - 1$.
Moreover, $P$ is a unique $M$-augmenting path due to the uniqueness of $M^*$.
Therefore, $M$ is $k$-maximal.

\begin{theorem}\label{thm:k-MM}
    For each integer $k \ge 1$, there is an $n^{O(k)}$-time algorithm for determining whether an input $n$-vertex graph has at least two $k$-maximal matchings.
\end{theorem}

\subsection{Enumerating local optimal solutions via Courcelle's theorem}\label{sec:cw}
For another tractable case, we would like to mention that a variant of Courcelle's theorem allows us to enumerate all $k$-maximal independent sets on bounded-cliquewidth graphs (and hence bounded-treewidth graphs) in polynomial delay for any fixed~$k$.

The property of being a 2-maximal independent set of a graph can be expressed by an MSO$_1$ formula (see \cite{Aoihon} for the syntax of MSO$_1$ formulas).
For a vertex set $X$, let
\begin{align*}
    \varphi(X) &= \texttt{ind}(X) \land \texttt{mxml(X)} \land \texttt{2-mxml}(X)\\
    \texttt{ind}(X) &= \forall_{u, v}(E(u, v) \rightarrow \neg(X(u) \land X(v))) \\
    \texttt{mxml}(X) &= \forall_{u}(\neg X(u) \rightarrow \neg \texttt{ind}(X \cup \{u\}))\\
    \texttt{2-mxml}(X) &= \forall_{u}(X(u) \rightarrow \neg\exists_{v, w}(X(v) \land X(w) \land \texttt{ind}((X \setminus \{u\}) \cup \{v, w\}))),
\end{align*}
where $X(u)$ and $E(u, v)$ are atomic formulas that are true if and only if $u \in X$ and $\{u, v\} \in E(G)$, respectively.\footnote{Precisely speaking, operators $\cup$ and $\setminus$ cannot be used in MSO$_1$ formulas but we can implement them within the syntax of those formulas.}
Then, $G \models \varphi(X)$ if and only if $X$ is a 2-maximal independent set of $G$.
By the enumeration version of Courcelle's theorem~\cite{Bagan06,Courcelle09}, we can enumerate all 2-maximal independent sets of bounded-cliquewidth graphs in polynomial delay, meaning that we can find a polynomial number of them (if they exist) in polynomial time.
Similar consequences for $k$-maximal independent sets, $k$-minimal dominating sets, and $k$-minimal feedback vertex sets are immediately obtained for all fixed~$k$ as those properties can be expressed by formulas in MSO$_1$ as well.

\begin{proposition}\label{prop:cw}
    There are polynomial-time algorithms that decide whether an input graph has at least two $k$-maximal independent sets, $k$-minimal dominating sets, or $k$-minimal feedback vertex sets.
\end{proposition}

\section{Local search versions of \textsc{Max (NAE)SAT} and \textsc{Max Cut}}
In this section, we focus on the local search versions of \textsc{Max SAT} and \textsc{Max Cut} with FLIP neighborhood~\cite{SchafferY91}.
For a CNF formula $\varphi$ and truth assignments $\alpha, \alpha'$ of $\varphi$, we say that $\alpha$ is \emph{adjacent} to $\alpha'$ in the FLIP neighborhood if $\alpha' = \alpha_x$ for some $x \in V(\varphi)$.
Similarly, for a graph $G$ and two cuts $C = \{X, Y\}$ and $C' = \{X', Y'\}$, we say that $C$ is \emph{adjacent} to $C'$ in the FLIP neighborhood if $X' = X \symdif \{v\}$ and $Y' = Y \symdif \{v\}$ for some $v \in V(G)$. 
Under these definitions, a truth assignment of $\varphi$ is locally optimal if it is unflippable and a cut of $G$ is locally optimal if it is stable.
We show that the problems of finding multiple unflippable assignments and multiple stable cuts are NP-hard.
To this end, we start with \textsc{Max NAESAT}.

We first observe that every 3-CNF formula has at least two NAE-unflippable assignments.
To see this, let $\alpha$ be a truth assignment of $\varphi$ that maximizes the number of NAE-satisfied clauses in $\varphi$.
This assignment is indeed NAE-unflippable, as it is a global optimum for \textsc{Max NAESAT}, and its complement $\overline{\alpha}$ is also NAE-unflippable and satisfies the same set of clauses.
The following theorem suggests, however, that finding a third NAE-unflippable assignment is hard.

\begin{theorem}\label{thm:sat:naesat}
    It is NP-complete to determine whether an input 3-CNF formula has at least three NAE-unflippable assignments.
\end{theorem}

The proof of \Cref{thm:sat:naesat} is also done by performing a polynomial-time reduction from \textsc{Maximal Independent Set Extension}.
Recall that, in \textsc{Maximal Independent Set Extension}, we are given a graph $G$ with vertex set $X$ and asked whether $G$ has a maximal independent set $D$ such that $D \cap X = \emptyset$.
Let us note that the isolated vertices in $V(G) \setminus X$ are contained in any maximal independent set $D$. 
By adding sufficiently many isolated vertices to $G$, we can assume that $|D| \ge |V(G) \setminus D|$ holds for every maximal independent set $D$ with $D \cap X = \emptyset$.

We construct a 3-CNF formula $\varphi$ as follows.
Let $m$ be the number of edges in $G$.
For each vertex $v \in V$, we associate a variable $x_v$, where the intention is that $x_v$ being true indicates that vertex~$v$ is included in an independent set.
We also use two additional variables $x^*$ and $y^*$, and a variable $s_e$ associated with each edge $e \in E(G)$.
For each edge $e = \{u, v\} \in E(G)$, we add clauses $(x^* \lor y^* \lor s_e)^2$, $(\neg s_e \lor \neg x_u \lor \neg x_v)^2$, and $(s_e \lor x^*)^3$ to $\varphi$.
Here, for a clause $c$ and a positive integer $t$, $c^t$ indicates the conjunction of $t$ copies of $c$.
For each vertex $v \in V(G)$, we add a clause $(x_v \lor y^*)$ and for each $v \in X$, 
add clauses $(x_v \lor \neg x^* \lor y^*)^{4m}$ and $(x^* \lor y^*)^{4m|X|}$.
The entire formula $\varphi$ is defined as
\begin{align*}
    \varphi &= \bigwedge_{e \in E(G)} \varphi_e \land \bigwedge_{v \in V(G)} \varphi_v \land \bigwedge_{v \in X} \varphi'_v \land (x^* \lor y^*)^{4m|X|};\\
    \varphi_e &= (x^* \lor y^* \lor s_e)^2 \land (\neg s_e \lor \neg x_u \lor \neg x_v)^2 \land (s_e \lor x^*)^3;\\
    \varphi_v &= (x_v \lor y^*);\\
    \varphi'_v &= (x_v \lor \neg x^* \lor y^*)^{4m},
\end{align*}
where $e = \{u, v\}$ in the second subformula $\varphi_e$.

There are two types of truth assignments $\alpha$ with (1) $\alpha(x^*) \neq \alpha(y^*)$ and (2) $\alpha(x^*) = \alpha(y^*)$.
We first show that there are only two NAE-unflippable assignments of type~(1).

\begin{lemma}\label{lem:naesat:type1}
    There are exactly two NAE-unflippable assignments $\alpha$ of $\varphi$ with $\alpha(x^*) \neq \alpha(y^*)$.
    Moreover, these two assignments are complements of each other.
\end{lemma}
\begin{proof}
    Suppose that $\alpha(x^*) = \True$ and $\alpha(y^*) = \False$.
    We first observe that $\alpha(s_e) = \False$ for all $e \in E(G)$.
    To see this, suppose that $\alpha(s_e) = \True$ for some $e \in E(G)$.
    Since clauses $(s_e \lor x^*)^3$ are not NAE-satisfied under $\alpha$ and clauses $(x^* \lor y^* \lor s_e)$ are still NAE-satisfied under $\alpha_{s_e}$, we have $\naesat(\alpha_{s_e}) - \naesat(\alpha) \ge 1$, meaning that $\alpha$ is NAE-flippable.
    Thus, $\alpha(s_e) = \False$ for all $e \in E(G)$.
    Moreover, for $v \in V(G)$ with $\alpha(x_v) = \False$, flipping the assignment of $x_v$ increases NAE-satisfied clauses by at least $4m - 2m + 1 = 2m + 1 \ge 1$, as $\varphi'_v$ becomes NAE-satisfied under $\alpha_{x_v}$.
    Hence, $\alpha(x_v) = \True$ for all $v$, and there is a unique NAE-unflippable assignment $\alpha$ with $\alpha(x^*) = \True$ and $\alpha(y^*) = \False$.
    Considering the symmetric case $\alpha(x^*) = \False$ and $\alpha(y^*) = \True$, the lemma holds.
\end{proof}

Now, we show that $G$ has a maximal independent set $D$ avoiding $X$ if and only if $\varphi$ has an NAE-unflippable assignment $\alpha$ with $\alpha(x^*) = \alpha(y^*)$.

\begin{lemma}\label{lem:nae:type2}
    There is a maximal independent set $D$ with $D \cap X = \emptyset$ in $G$ if and only if there is an NAE-unflippable assignment $\alpha$ with $\alpha(x^*) = \alpha(y^*)$ for $\varphi$.
\end{lemma}
\begin{proof}
    Suppose that $G$ has a maximal independent set $D$ with $D \cap X = \emptyset$.
    We define a truth assignment $\alpha$ for $\varphi$ as:
    \begin{align*}
        \alpha(x) =
        \begin{cases}
        \True & x \in \{x_v : v \in D\} \cup \{s_e : e \in E(G)\}\\
        \False & \text{otherwise}
        \end{cases}.
    \end{align*}
    Note that all the clauses of the form $(x^* \lor y^* \lor s_e)$, $(\neg s_e \lor \neg x_u \lor \neg x_v)$, $(s_e \lor x^*)$, and $(x_v \lor \neg x^* \lor y^*)$ are NAE-satisfied under $\alpha$, as $D$ is an independent set of $G$.
    We show that $\alpha$ is NAE-unflippable at each variable.
    For $e \in E(G)$, when flipping the assignment of $s_e$ to $\False$, it does not increase NAE-satisfied clauses at all but makes clauses $(x^* \lor y^* \lor s_e)$ and $(s_e \lor x^*)$ NAE-unsatisfied.
    For $u \in D$, $\alpha$ is NAE-unflippable at $x_u$ since all the clauses appearing $x_u$ are NAE-satisfied under $\alpha$.
    For $u \in V(G) \setminus D$, there is a neighbor $v \in N_G(u)$ of $u$ that belongs to $D$ due to the maximality of $D$.
    If we flip the assignment of $x_u$ to $\True$, clause~$(x_u \lor y^*)$ becomes NAE-satisfied but clauses~$(\neg s_e \lor \neg x_u \lor \neg x_v)^2$ become NAE-unsatisfied and clause~$(x_u \lor \neg x^* \lor y^*)$ remains NAE-satisfied.
    Thus, flipping the assignment of $x_u$ does not increase NAE-satisfied clauses, yielding $\alpha$ is NAE-unflippable at $x_u$ for all $u \in V(G) \setminus D$.
    Finally, we conclude that $\alpha$ is NAE-unflippable at $x^*$ and $y^*$: all $4m|X|$ clauses of $(x^* \lor y^*)$ become NAE-satisfied but all $4m$ clauses of $(x_v \lor \neg x^* \lor y^*)$ become NAE-unsatisfied for $v \in X$ under $\alpha_{x^*}$;
    the increment in the number of NAE-satisfied clauses is determined by the clauses in $\bigwedge_{v\in V(G)}(x_v \lor y^*)$, which is 
    \begin{align*}
        &|\{v \in V(G) : \alpha(x_v) = \False\}| - |\{v \in V(G) : \alpha(x_v) = \True\}| = |V(G) \setminus D| - |D| \le 0
    \end{align*}
    as $|D| \ge |V(G) \setminus D|$.
    Therefore, $\alpha$ is NAE-unflippable.

    Conversely, suppose that there is an NAE-unflippable assignment $\alpha$ with $\alpha(x^*) = \alpha(y^*)$ for $\varphi$.
    Without loss of generality, we assume that $\alpha(x^*) = \alpha(y^*) = \False$ as the complement of $\alpha$ is also an NAE-unflippable assignment.
    Note that all the clauses of the form $(x_v \lor \neg x^* \lor y^*)$ are NAE-satisfied regardless of the assignment of $x_v$.
    Let $D = \{v \in V(G) : \alpha(x_v) = \True\}$.
    We then show that $D$ is a maximal independent set of $G$ with $D \cap X = \emptyset$.

    We first observe that $\alpha(s_e) = \True$ for all $e \in E(G)$.
    To see this, suppose otherwise that $\alpha(s_e) = \False$ for some $e \in E(G)$.
    As $\alpha(x^*) = \alpha(y^*)= \False$, flipping the assignment of $s_e$ would make both $(x^* \lor y^* \lor s_e)^2$ and $(s_e \lor x^*)^3$ NAE-satisfied, but $(\neg s_e \lor \neg x_u \lor \neg x_v)^2$ can be NAE-unsatisfied.
    This implies that it increases the number of NAE-satisfied clauses by at least~3, contradicting the assumption that $\alpha$ is NAE-unflippable.
    We next show that $D$ is an independent set of $G$.
    Suppose that there is an edge $e = \{u, v\}$ such that $u, v \in D$.
    Then, clauses~$(\neg s_e \lor \neg x_u \lor \neg x_v)^2$ are NAE-unsatisfied. 
    This implies that flipping the assignment $\alpha(x_u)$ to $\False$ would make these two clauses NAE-satisfied and clause~$(x_v \lor y^*)$ NAE-unsatisfied, meaning that $\alpha$ is NAE-flippable.
    To see the maximality of $D$, suppose that there is a vertex $u \in V(G) \setminus D$ such that $D \cup \{u\}$ is an independent set of $G$.
    This implies that $\alpha(x_u) = \False$.
    When flipping its assignment to $\True$, clause~$(x_v \lor y^*)$ becomes NAE-satisfied and other clauses remain unchanged as $N(u) \cap D = \emptyset$, a contradiction.
    Finally, we show that $D \cap X = \emptyset$.
    Suppose that there is $u \in D \cap X$.
    Since $\alpha(x_u) \neq \alpha(y^*)$ and $\alpha(x^*) = \alpha(y^*)$, flipping the assignment of $x^*$ would make clauses $(x^* \lor y^*)^{4m|X|}$ NAE-satisfied, clauses $(s_e \lor x^*)^3$ NAE-unsatisfied for $e \in E(G)$, and clauses $(x_v \lor \neg x^* \lor y^*)^{4m}$ NAE-unsatisfied for $v \in X \setminus D$. 
    Thus, it increases NAE-satisfied clauses at least $4m|X| - 4m(|X| - 1) - 3m > 0$, meaning that $\alpha$ is NAE-flippable.

    Therefore, $D$ is a maximal independent set of $G$ with $D \cap X = \emptyset$.
\end{proof}

By \Cref{lem:naesat:type1}, there are exactly two NAE-unflippable assignments $\alpha$ of $\varphi$ with $\alpha(x^*) \neq \alpha(y^*)$ and, by \Cref{lem:nae:type2}, there is an NAE-unflippable assignment $\alpha$ of $\varphi$ with $\alpha(x^*) = \alpha(y^*)$ if and only if $G$ has a maximal independent set $D$ with $D \cap X = \emptyset$.
Therefore, we can determine the existence of $D$ by checking whether $\varphi$ has at least three NAE-unflippable assignments.
This proves \Cref{thm:sat:naesat}.

We next show that the problem remains hard even for positive CNF formulas, that is, each clause contains only positive literals.
The proof is almost analogous to the standard reduction from \textsc{NAE3SAT} to \text{Positive NAE3SAT}.
Let $\varphi$ be a CNF formula with $m$ clauses.
We replace each negative literal $\neg x$ with a fresh variable $x'$ and add $(x \lor x')^{k}$ for sufficiently large $k$, say $k = m + 1$, for each variable $x \in V(\varphi)$.
The obtained positive CNF formula is denoted by $\varphi'$.
The following lemma ensures that the new variable $x'$ serves as the negation of $x$ for any NAE-unflippable assignment.

\begin{lemma}\label{lem:sat:posnaesat}
    For every NAE-unflippable assignment $\alpha$ for $\varphi'$, $\alpha(x) \neq \alpha(x')$ for all $x \in V(\varphi)$. 
\end{lemma}
\begin{proof}
    Suppose that there is $x \in V(\varphi)$ such that $\alpha(x) = \alpha(x')$.
    By flipping the assignment of $x'$, clause~$(x \lor x')$ becomes NAE-satisfied.
    Since $\varphi'$ contains $m + 1$ copies of this clause, we have $\naesat(\alpha_{x'}) - \naesat(\alpha) \ge 1$.
    This contradicts the fact that $\alpha$ is NAE-unflippable.
\end{proof}

The above lemma implies that there is a bijection between the set of NAE-unflippable assignments for $\varphi$ and those for $\varphi'$, proving the following theorem.

\begin{theorem}\label{thm:sat:posnaesat}
    It is NP-complete to determine whether an input positive 3-CNF formula has at least three NAE-unflippable assignments.
\end{theorem}

Now, we turn our attention to \textsc{Max Cut}.
Let $\varphi$ be a positive 3-CNF formula.
As we will perform a reduction from \textsc{Max NAE3SAT}, we can assume that $\varphi$ has no unit clauses, as they are always NAE-unsatisfied.
From $\varphi$, we construct a multigraph $G$ as follows. 
The vertex set of $G$ corresponds to $V(\varphi)$.
For each clause with two literals, say $(x \lor y)$, we add two parallel edges between $x$ and $y$, and for each clause with three literals, say $(x \lor y \lor z)$, we add edges $\{x, y\}, \{y, z\}, \{z, x\}$, forming a triangle.
An \emph{ordered} cut of $G$ is an ordered pair $(X, Y)$ such that $X \cup Y = V(G)$ and $X \cap Y = \emptyset$.
An ordered cut $(X, Y)$ is said to be stable if its unordered counterpart $\{X, Y\}$ is stable.

\begin{lemma}\label{lem:stablecut-bij}
    There is a bijection between the set of NAE-unflippable assignments for $\varphi$ and the set of ordered stable cuts of $G$.
\end{lemma}
\begin{proof}
    From a truth assignment $\alpha$ of $\varphi$, we can naturally define an ordered cut $(X, Y)$ as $X = \{x \in V(\varphi) : \alpha(x) = \True\}$ and $Y = \{x \in V(\varphi) : \alpha(x) = \False\}$.
    Conversely, we can define a truth assignment of $\varphi$ from an ordered cut of $G$ in the same way as above.
    Note that the complement of $\alpha$ defines the ordered cut $(Y, X)$ and vice-versa.
    Now, we show that $\alpha$ is NAE-unflippable if and only if $(X, Y)$ is stable.

    Suppose that $\alpha$ is NAE-unflippable.
    Observe that when a clause $c$ is NAE-satisfied, the variables contained in $c$ contributes exactly $2$ to the cut $(X, Y)$.
    Thus, for $x \in X \cup Y$, 
    \begin{align*}
        |E(X, Y)| - |E(X \symdif \{x\}, Y \symdif \{x\})| = 2\cdot\naesat(\alpha) - 2\cdot\naesat(\alpha_x) \ge 0,
    \end{align*}
    implying that $(X, Y)$ is stable.
    This implication is reversible, proving the claim of the lemma.
\end{proof}

It is easy to see that $G$ has at least three ordered stable cuts if and only if it has at least two (unordered) stable cuts.
By~\Cref{thm:sat:posnaesat} and \Cref{lem:stablecut-bij}, the following theorem holds. 

\begin{theorem}\label{thm:maxcut}
    It is NP-complete to determine whether an input multigraph $G$ has at least two stable cuts.
\end{theorem}

Finally, we consider \textsc{Max 2SAT}.
The proof is almost the same as a standard reduction from \textsc{Max Cut} to \textsc{Max 2SAT}.

\begin{theorem}\label{thm:maxsat}
    It is NP-complete to determine whether an input 2-CNF formula has at least two unflippable assignments.
\end{theorem}
\begin{proof}
    We perform a reduction from the local search version of \textsc{Max Cut} to the local search version of \text{Max 2SAT}.
    Let $G$ be a multigraph.
    We choose an arbitrary vertex and denote it by $v^*$.
    We define a 2-CNF formula $\varphi$ with $V(\varphi) = V(G)$ as:
    \begin{align*}
        \bigwedge_{\{u, v\} \in E(G)} \left( (u \lor v) \land (\neg u \lor \neg v) \right) \land (v^*)^{2|E(G)| + 1}.
    \end{align*}
    Note that for a truth assignment $\alpha$ of $\varphi$ and for an edge $\{u, v\} \in E(G)$, either the two clauses corresponding to it are both satisfied (when $\alpha(u) \neq \alpha(v)$) or exactly one of them is satisfied (when $\alpha(u) = \alpha(v)$).
    
    Let $\{X, Y\}$ be a cut of $G$ with $v^* \in X$.
    Then, define a truth assignment of $\varphi$ as
    \begin{align*}
        \alpha(u) = \begin{cases}
            \True & u \in X\\
            \False & u \in Y
        \end{cases}
    \end{align*}
    for each $u \in V(G) \setminus \{v^*\}$.
    From a truth assignment of $\varphi$, we can define a cut in the same way as above.

    Let $\alpha$ be an unflippable assignment of $\varphi$.
    Then, it follows that $\alpha(v^*) = \True$ as $\varphi$ contains $2|E(G)| + 1$ copies of unit clause~$(v^*)$.
    For $u \in V(\varphi) \setminus \{v^*\}$, we have
    \begin{align*}
        \sat(\alpha) - \sat(\alpha_u) &= 2|E(X, Y)| + |E(G) \setminus E(X, Y)|\\
        &\hspace{1cm} - 2|E(X \symdif \{u\}, Y \symdif \{u\})| - |E(G) \setminus E(X \symdif \{u\}, Y \symdif \{u\})|\\
        &= |E(X, Y)| + |E(G)| - |E(X \symdif \{u\}, Y \symdif \{u\})| - |E(G)|\\
        &= |E(X, Y)| - |E(X \symdif \{u\}, Y \symdif \{u\})|.
    \end{align*}
    As $\sat(\alpha) \ge \sat(\alpha_u)$, the cut $\{X, Y\}$ is stable.

    Conversely, from a stable cut $\{X, Y\}$ with $v^* \in X$, we can conclude that the truth assignment $\alpha$ is unflippable, by considering the above equalities in the reverse direction.
\end{proof}

\section{Concluding remarks}
In this paper, we initiated the study of the complexity of finding multiple local optima in unweighted combinatorial optimization problems.
Our results suggest that there are several natural local search problems for which one of the local optima is easy to find, but two (or three) of them are hard to find, and give rise to several interesting future directions.
\begin{itemize}
    \item In \Cref{thm:maxcut} and \Cref{thm:maxsat}, the constructed graphs and 2-CNF formulas have (unweighted) multiedges and multiset of clauses, respectively.
    Since these multiple objects can be encoded by a polynomially-weighted single object, these results also hold for simple polynomially-weighted graphs and formulas.
    It would be interesting to investigate that the problems are still hard even for unweighted simple graphs and formulas.
    
    \item The primal goal of this paper is to establish that, even on very simple neighborhood structures, finding multiple local optima is computationally intractable in natural combinatorial optimization problems.
    To extend our results, there are various local search problems with other neighborhood structures studied in the context of PLS-completeness.
    We believe that these local search problems are also intractable in our setting, while it would be interesting to investigate natural local search problems that are PLS-complete but tractable in our setting.
    
    \item We show that finding two $k$-maximal independents can be solved in polynomial time when the input graph is restricted to the class of line graphs (\Cref{thm:k-MM}) or to that of bounded-cliquewidth graphs (\Cref{prop:cw}).
    However, these results might be of no importance when it comes to finding a global optimum for \textsc{Independent Set}, as we can find a global one in polynomial time on these classes.
    A natural question is to investigate a class of graphs on which \textsc{Independent Set} is NP-hard but where finding multiple 2-maximal independent sets is easy, while it is also an interesting question whether the opposite situation can occur: finding a global optimum is easy, but finding multiple local optima is hard.
\end{itemize}

\printbibliography

@inproceedings{CaselFGMS19,
  author       = {Katrin Casel and
                  Henning Fernau and
                  Mehdi Khosravian Ghadikolaei and
                  J{\'{e}}r{\^{o}}me Monnot and
                  Florian Sikora},
  editor       = {Pinar Heggernes},
  title        = {Extension of Vertex Cover and Independent Set in Some Classes of Graphs},
  booktitle    = {Proceedings of {CIAC} 2019},
  series       = {Lecture Notes in Computer Science},
  volume       = {11485},
  pages        = {124--136},
  publisher    = {Springer},
  year         = {2019},
  url          = {https://doi.org/10.1007/978-3-030-17402-6_11},
  doi          = {10.1007/978-3-030-17402-6_11},
  bibsource    = {dblp computer science bibliography, https://dblp.org}
}

@article{JohnsonPY88,
  author       = {David S. Johnson and
                  Christos H. Papadimitriou and
                  Mihalis Yannakakis},
  title        = {How Easy is Local Search?},
  journal      = {J. Comput. Syst. Sci.},
  volume       = {37},
  number       = {1},
  pages        = {79--100},
  year         = {1988},
  url          = {https://doi.org/10.1016/0022-0000(88)90046-3},
  doi          = {10.1016/0022-0000(88)90046-3},
  bibsource    = {dblp computer science bibliography, https://dblp.org}
}

@article{KomusiewiczM24,
    author = {Komusiewicz, Christian and Morawietz, Nils},
    title = {Finding 3-Swap-Optimal Independent Sets and Dominating Sets is Hard},
    year = {2024},
    journal = {ACM Transactions on Computation Theory},
    publisher = {Association for Computing Machinery},
    address = {New York, NY, USA},
    issn = {1942-3454},
    url = {https://doi.org/10.1145/3700642},
    doi = {10.1145/3700642},
}

@article{MarxS11,
  author       = {D{\'{a}}niel Marx and
                  Ildik{\'{o}} Schlotter},
  title        = {Stable assignment with couples: Parameterized complexity and local
                  search},
  journal      = {Discret. Optim.},
  volume       = {8},
  number       = {1},
  pages        = {25--40},
  year         = {2011},
  url          = {https://doi.org/10.1016/j.disopt.2010.07.004},
  doi          = {10.1016/J.DISOPT.2010.07.004},
  bibsource    = {dblp computer science bibliography, https://dblp.org}
}

@inproceedings{KomusiewiczM22,
  author       = {Christian Komusiewicz and
                  Nils Morawietz},
  editor       = {Holger Dell and
                  Jesper Nederlof},
  title        = {Parameterized Local Search for Vertex Cover: When Only the Search
                  Radius Is Crucial},
  booktitle    = {Proceedings of {IPEC} 2022},
  series       = {LIPIcs},
  volume       = {249},
  pages        = {20:1--20:18},
  year         = {2022},
  url          = {https://doi.org/10.4230/LIPIcs.IPEC.2022.20},
  doi          = {10.4230/LIPICS.IPEC.2022.20},
  bibsource    = {dblp computer science bibliography, https://dblp.org}
}

@article{GuoHNS13,
  author       = {Jiong Guo and
                  Sepp Hartung and
                  Rolf Niedermeier and
                  Ondrej Such{\'{y}}},
  title        = {The Parameterized Complexity of Local Search for TSP, More Refined},
  journal      = {Algorithmica},
  volume       = {67},
  number       = {1},
  pages        = {89--110},
  year         = {2013},
  url          = {https://doi.org/10.1007/s00453-012-9685-8},
  doi          = {10.1007/S00453-012-9685-8},
  bibsource    = {dblp computer science bibliography, https://dblp.org}
}

@article{BergBJW21,
  author       = {Mark de Berg and
                  Kevin Buchin and
                  Bart M. P. Jansen and
                  Gerhard J. Woeginger},
  title        = {Fine-grained Complexity Analysis of Two Classic {TSP} Variants},
  journal      = {{ACM} Trans. Algorithms},
  volume       = {17},
  number       = {1},
  pages        = {5:1--5:29},
  year         = {2021},
  url          = {https://doi.org/10.1145/3414845},
  doi          = {10.1145/3414845},
  bibsource    = {dblp computer science bibliography, https://dblp.org}
}

@article{Marx08,
  author       = {D{\'{a}}niel Marx},
  title        = {Searching the $k$-change neighborhood for {TSP} is W[1]-hard},
  journal      = {Oper. Res. Lett.},
  volume       = {36},
  number       = {1},
  pages        = {31--36},
  year         = {2008},
  doi          = {10.1016/J.ORL.2007.02.008},
}

@inproceedings{GarvardtGKM23,
  author       = {Jaroslav Garvardt and
                  Niels Gr{\"{u}}ttemeier and
                  Christian Komusiewicz and
                  Nils Morawietz},
  title        = {Parameterized Local Search for Max $c$-Cut},
  booktitle    = {Proceedings of {IJCAI} 2023},
  pages        = {5586--5594},
  publisher    = {ijcai.org},
  year         = {2023},
  url          = {https://doi.org/10.24963/ijcai.2023/620},
  doi          = {10.24963/IJCAI.2023/620},
}

@article{FellowsFLRSV12,
  author       = {Michael R. Fellows and
                  Fedor V. Fomin and
                  Daniel Lokshtanov and
                  Frances A. Rosamond and
                  Saket Saurabh and
                  Yngve Villanger},
  title        = {Local search: Is brute-force avoidable?},
  journal      = {J. Comput. Syst. Sci.},
  volume       = {78},
  number       = {3},
  pages        = {707--719},
  year         = {2012},
  url          = {https://doi.org/10.1016/j.jcss.2011.10.003},
  doi          = {10.1016/J.JCSS.2011.10.003},
  bibsource    = {dblp computer science bibliography, https://dblp.org}
}

@inproceedings{GaspersKOSS12,
  author       = {Serge Gaspers and
                  Eun Jung Kim and
                  Sebastian Ordyniak and
                  Saket Saurabh and
                  Stefan Szeider},
  title        = {Don't Be Strict in Local Search!},
  booktitle    = {Proceedings of {AAAI} 2012},
  pages        = {486--492},
  publisher    = {{AAAI} Press},
  year         = {2012},
  url          = {https://doi.org/10.1609/aaai.v26i1.8128},
  doi          = {10.1609/AAAI.V26I1.8128},
  bibsource    = {dblp computer science bibliography, https://dblp.org}
}

@article{Szeider11,
  author       = {Stefan Szeider},
  title        = {The parameterized complexity of $k$-flip local search for {SAT} and
                  {MAX} {SAT}},
  journal      = {Discret. Optim.},
  volume       = {8},
  number       = {1},
  pages        = {139--145},
  year         = {2011},
  url          = {https://doi.org/10.1016/j.disopt.2010.07.003},
  doi          = {10.1016/J.DISOPT.2010.07.003},
  bibsource    = {dblp computer science bibliography, https://dblp.org}
}

@article{TsukiyamaIAS77,
  author       = {Shuji Tsukiyama and
                  Mikio Ide and
                  Hiromu Ariyoshi and
                  Isao Shirakawa},
  title        = {A New Algorithm for Generating All the Maximal Independent Sets},
  journal      = {{SIAM} J. Comput.},
  volume       = {6},
  number       = {3},
  pages        = {505--517},
  year         = {1977},
  url          = {https://doi.org/10.1137/0206036},
  doi          = {10.1137/0206036},
  bibsource    = {dblp computer science bibliography, https://dblp.org}
}

@inproceedings{MonienDT10,
  author       = {Burkhard Monien and
                  Dominic Dumrauf and
                  Tobias Tscheuschner},
  editor       = {Samson Abramsky and
                  Cyril Gavoille and
                  Claude Kirchner and
                  Friedhelm Meyer auf der Heide and
                  Paul G. Spirakis},
  title        = {Local Search: Simple, Successful, But Sometimes Sluggish},
  booktitle    = {Proceedings of the 37th International Colloquium on Automata, Languages and Programming, {ICALP} 2010, Part {I}},
  series       = {Lecture Notes in Computer Science},
  volume       = {6198},
  pages        = {1--17},
  publisher    = {Springer},
  year         = {2010},
  url          = {https://doi.org/10.1007/978-3-642-14165-2_1},
  doi          = {10.1007/978-3-642-14165-2_1},
  bibsource    = {dblp computer science bibliography, https://dblp.org}
}

@article{SchafferY91,
  author       = {Alejandro A. Sch{\"{a}}ffer and
                  Mihalis Yannakakis},
  title        = {Simple Local Search Problems That are Hard to Solve},
  journal      = {{SIAM} J. Comput.},
  volume       = {20},
  number       = {1},
  pages        = {56--87},
  year         = {1991},
  url          = {https://doi.org/10.1137/0220004},
  doi          = {10.1137/0220004},
  bibsource    = {dblp computer science bibliography, https://dblp.org}
}

@article{ChandraKT99,
  author       = {Barun Chandra and
                  Howard J. Karloff and
                  Craig A. Tovey},
  title        = {New Results on the Old k-opt Algorithm for the Traveling Salesman
                  Problem},
  journal      = {{SIAM} J. Comput.},
  volume       = {28},
  number       = {6},
  pages        = {1998--2029},
  year         = {1999},
  url          = {https://doi.org/10.1137/S0097539793251244},
  doi          = {10.1137/S0097539793251244},
  bibsource    = {dblp computer science bibliography, https://dblp.org}
}

@article{EnglertRV14,
  author       = {Matthias Englert and
                  Heiko R{\"{o}}glin and
                  Berthold V{\"{o}}cking},
  title        = {Worst Case and Probabilistic Analysis of the 2-Opt Algorithm for the
                  {TSP}},
  journal      = {Algorithmica},
  volume       = {68},
  number       = {1},
  pages        = {190--264},
  year         = {2014},
  url          = {https://doi.org/10.1007/s00453-013-9801-4},
  doi          = {10.1007/S00453-013-9801-4},
  bibsource    = {dblp computer science bibliography, https://dblp.org}
}

@inproceedings{Krentel89,
  author       = {Mark W. Krentel},
  title        = {Structure in Locally Optimal Solutions (Extended Abstract)},
  booktitle    = {30th Annual Symposium on Foundations of Computer Science, {FOCS} 1989},
  pages        = {216--221},
  publisher    = {{IEEE} Computer Society},
  year         = {1989},
  url          = {https://doi.org/10.1109/SFCS.1989.63481},
  doi          = {10.1109/SFCS.1989.63481},
  bibsource    = {dblp computer science bibliography, https://dblp.org}
}

@incollection{LourencoMS03,
  author       = {Helena R. Louren{\c{c}}o and
                  Olivier C. Martin and
                  Thomas St{\"{u}}tzle},
  editor       = {Fred W. Glover and
                  Gary A. Kochenberger},
  title        = {Iterated Local Search},
  booktitle    = {Handbook of Metaheuristics},
  series       = {International Series in Operations Research {\&} Management Science},
  volume       = {57},
  pages        = {320--353},
  publisher    = {Kluwer / Springer},
  year         = {2003},
  url          = {https://doi.org/10.1007/0-306-48056-5_11},
  doi          = {10.1007/0-306-48056-5_11},
  bibsource    = {dblp computer science bibliography, https://dblp.org}
}

@book{GK2003,
  editor       = {Fred W. Glover and
                  Gary A. Kochenberger},
  title        = {Handbook of Metaheuristics},
  series       = {International Series in Operations Research {\&} Management Science},
  volume       = {57},
  publisher    = {Kluwer / Springer},
  year         = {2003},
  url          = {https://doi.org/10.1007/b101874},
  doi          = {10.1007/B101874},
  isbn         = {978-1-4020-7263-5},
  bibsource    = {dblp computer science bibliography, https://dblp.org}
}

@book{2018heuristics,
  editor       = {Rafael Mart{\'{i}} and
                  Panos M. Pardalos and
                  Mauricio G. C. Resende},
  title        = {Handbook of Heuristics},
  publisher    = {Springer},
  year         = {2018},
  url          = {https://doi.org/10.1007/978-3-319-07124-4},
  doi          = {10.1007/978-3-319-07124-4},
  isbn         = {978-3-319-07123-7},
  bibsource    = {dblp computer science bibliography, https://dblp.org}
}

@article{JohnsonP88,
  author       = {David S. Johnson and
                  Christos H. Papadimitriou and
                  Mihalis Yannakakis},
  title        = {On Generating All Maximal Independent Sets},
  journal      = {Inf. Process. Lett.},
  volume       = {27},
  number       = {3},
  pages        = {119--123},
  year         = {1988},
  url          = {https://doi.org/10.1016/0020-0190(88)90065-8},
  doi          = {10.1016/0020-0190(88)90065-8},
  bibsource    = {dblp computer science bibliography, https://dblp.org}
}

@article{MartiRR13,
  author       = {Rafael Mart{\'{i}} and
                  Mauricio G. C. Resende and
                  Celso C. Ribeiro},
  title        = {Multi-start methods for combinatorial optimization},
  journal      = {Eur. J. Oper. Res.},
  volume       = {226},
  number       = {1},
  pages        = {1--8},
  year         = {2013},
  url          = {https://doi.org/10.1016/j.ejor.2012.10.012},
  doi          = {10.1016/J.EJOR.2012.10.012},
  bibsource    = {dblp computer science bibliography, https://dblp.org}
}

@inproceedings{Bagan06,
  author       = {Guillaume Bagan},
  editor       = {Zolt{\'{a}}n {\'{E}}sik},
  title        = {{MSO} Queries on Tree Decomposable Structures Are Computable with
                  Linear Delay},
  booktitle    = {Proceedings of the 20th International Workshop on Computer Science Logic, {CSL} 2006},
  series       = {Lecture Notes in Computer Science},
  volume       = {4207},
  pages        = {167--181},
  publisher    = {Springer},
  year         = {2006},
  url          = {https://doi.org/10.1007/11874683_11},
  doi          = {10.1007/11874683_11},
}

@article{Courcelle09,
  author       = {Bruno Courcelle},
  title        = {Linear delay enumeration and monadic second-order logic},
  journal      = {Discret. Appl. Math.},
  volume       = {157},
  number       = {12},
  pages        = {2675--2700},
  year         = {2009},
  url          = {https://doi.org/10.1016/j.dam.2008.08.021},
  doi          = {10.1016/J.DAM.2008.08.021},
}

@book{Aoihon,
  author    = {Marek Cygan and
               Fedor V. Fomin and
               {\L}ukasz Kowalik and
               Daniel Lokshtanov and
               D{\'{a}}niel Marx and
               Marcin Pilipczuk and
               Micha{\l} Pilipczuk and
               Saket Saurabh},
  title        = {Parameterized Algorithms},
  publisher    = {Springer},
  year         = {2015},
  doi          = {10.1007/978-3-319-21275-3},
  isbn         = {978-3-319-21274-6},
  bibsource    = {dblp computer science bibliography, https://dblp.org}
}

@article{ConteT22,
  author       = {Alessio Conte and
                  Etsuji Tomita},
  title        = {On the overall and delay complexity of the {CLIQUES} and Bron-Kerbosch
                  algorithms},
  journal      = {Theor. Comput. Sci.},
  volume       = {899},
  pages        = {1--24},
  year         = {2022},
  doi          = {10.1016/J.TCS.2021.11.005},
  bibsource    = {dblp computer science bibliography, https://dblp.org}
}

@article{GabowKT01,
  author       = {Harold N. Gabow and
                  Haim Kaplan and
                  Robert Endre Tarjan},
  title        = {Unique Maximum Matching Algorithms},
  journal      = {J. Algorithms},
  volume       = {40},
  number       = {2},
  pages        = {159--183},
  year         = {2001},
  url          = {https://doi.org/10.1006/jagm.2001.1167},
  doi          = {10.1006/JAGM.2001.1167},
  bibsource    = {dblp computer science bibliography, https://dblp.org}
}

@inproceedings{SelmanKC94,
  author       = {Bart Selman and
                  Henry A. Kautz and
                  Bram Cohen},
  editor       = {Barbara Hayes{-}Roth and
                  Richard E. Korf},
  title        = {Noise Strategies for Improving Local Search},
  booktitle    = {Proceedings of the 12th National Conference on Artificial Intelligence,
                  Seattle, WA, USA, July 31 - August 4, 1994, Volume 1},
  pages        = {337--343},
  publisher    = {{AAAI} Press / The {MIT} Press},
  year         = {1994},
  url          = {http://www.aaai.org/Library/AAAI/1994/aaai94-051.php},
  bibsource    = {dblp computer science bibliography, https://dblp.org}
}

@inproceedings{SelmanLM92,
  author       = {Bart Selman and
                  Hector J. Levesque and
                  David G. Mitchell},
  editor       = {William R. Swartout},
  title        = {A New Method for Solving Hard Satisfiability Problems},
  booktitle    = {Proceedings of the 10th National Conference on Artificial Intelligence,
                  San Jose, CA, USA, July 12-16, 1992},
  pages        = {440--446},
  publisher    = {{AAAI} Press / The {MIT} Press},
  year         = {1992},
  url          = {http://www.aaai.org/Library/AAAI/1992/aaai92-068.php},
  bibsource    = {dblp computer science bibliography, https://dblp.org}
}

@inproceedings{SunWWL0Y24,
  author       = {Rui Sun and
                  Yiyuan Wang and
                  Shimao Wang and
                  Hui Li and
                  Ximing Li and
                  Minghao Yin},
  title        = {Nukplex: An Efficient Local Search Algorithm for Maximum K-Plex Problem},
  booktitle    = {Proceedings of the Thirty-Third International Joint Conference on
                  Artificial Intelligence, {IJCAI} 2024, Jeju, South Korea, August 3-9,
                  2024},
  pages        = {7029--7037},
  publisher    = {ijcai.org},
  year         = {2024},
  url          = {https://www.ijcai.org/proceedings/2024/777},
  bibsource    = {dblp computer science bibliography, https://dblp.org}
}

@inproceedings{ZhangL05,
  author       = {Weixiong Zhang and
                  Moshe Looks},
  editor       = {Leslie Pack Kaelbling and
                  Alessandro Saffiotti},
  title        = {A Novel Local Search Algorithm for the Traveling Salesman Problem
                  that Exploits Backbones},
  booktitle    = {IJCAI-05, Proceedings of the Nineteenth International Joint Conference
                  on Artificial Intelligence, Edinburgh, Scotland, UK, July 30 - August
                  5, 2005},
  pages        = {343--350},
  publisher    = {Professional Book Center},
  year         = {2005},
  url          = {http://ijcai.org/Proceedings/05/Papers/1443.pdf},
  bibsource    = {dblp computer science bibliography, https://dblp.org}
}

@book{WS11,
  author       = {David P. Williamson and
                  David B. Shmoys},
  title        = {The Design of Approximation Algorithms},
  publisher    = {Cambridge University Press},
  year         = {2011},
  url          = {http://www.cambridge.org/de/knowledge/isbn/item5759340/?site\_locale=de\_DE},
  isbn         = {978-0-521-19527-0},
  bibsource    = {dblp computer science bibliography, https://dblp.org}
}

@book{MichielsAK07,
  author       = {Wil Michiels and
                  Emile H. L. Aarts and
                  Jan H. M. Korst},
  title        = {Theoretical aspects of local search},
  series       = {Monographs in Theoretical Computer Science. An {EATCS} Series},
  publisher    = {Springer},
  year         = {2007},
  url          = {https://doi.org/10.1007/978-3-540-35854-1},
  doi          = {10.1007/978-3-540-35854-1},
  isbn         = {978-3-540-35853-4},
  bibsource    = {dblp computer science bibliography, https://dblp.org}
}

@book{aarts2003local,
  author    = {Emile H. L. Aarts and Jan Karel Lenstra},
  title     = {Local Search in Combinatorial Optimization},
  edition   = {2nd},
  publisher = {Princeton University Press},
  address   = {Princeton, NJ},
  year      = {2003},
  isbn      = {9780691115221},
}

@book{korte2018combinatorial,
  title     = {Combinatorial Optimization: Theory and Algorithms},
  author    = {Korte, Bernhard and Vygen, Jens},
  year      = {2018},
  edition   = {6},
  publisher = {Springer},
  series    = {Algorithms and Combinatorics},
  volume    = {21},
  isbn      = {978-3662560396},
}

\end{document}